\documentclass[amsmath]{article}
\usepackage{amsmath,amssymb,amsthm,mathrsfs,txfonts,epsfig}
\usepackage{cite}
\usepackage{latexsym,color,booktabs}
\usepackage[mathscr]{eucal}
\DeclareMathOperator{\tr}{tr}

\newtheorem{theorem}{Theorem}

\newtheorem{example}{Example}
\newtheorem{lemma}{Lemma}
\newtheorem{corollary}{Corollary}

\usepackage{CJK}
\begin{document}
\date{}
\title{\bf Quantum Fisher information-based detection of genuine tripartite entanglement}
\author{ \\Long-Mei Yang$^1$, Bao-Zhi Sun$^2$, Bin Chen$^3$, Shao-Ming Fei$^{4,5}$, Zhi-Xi Wang$^4$
\thanks{Corresponding author: wangzhx@cnu.edu.cn}
\\
{\footnotesize $^1$State Key Laboratory of Low-Dimensional Quantum Physics and Department of Physics, Tsinghua University, Beijing 100084, China}\\
{\footnotesize $^2$School of Mathematical Sciences, Qufu Normal University, Qufu 273165, China}\\
{\footnotesize $^3$College of Mathematical Science, Tianjin Normal University, Tianjin 300387, China}\\
{\footnotesize $^4$School of Mathematical Sciences, Capital Normal University, Beijing 100048, China}\\
{\footnotesize $^5$Max-Planck-Institute for Mathematics in the Sciences, 04103 Leipzig, Germany}}
\maketitle
\begin{abstract}
Genuine multipartite entanglement plays important roles in quantum information processing. The detection of genuine multipartite entanglement
has been long time a challenging problem in the theory of quantum entanglement.
We propose a criterion for detecting genuine tripartite entanglement of arbitrary dimensional tripartite states based on quantum Fisher information.
We show that this criterion is more effective for some states in detecting genuine tripartite entanglement by detailed examples.
\end{abstract}

\section{Introduction}
Multipartite entanglement, as one of the most remarkable resources in the theory of quantum information processing and quantum computation, has been investigated extensively in the last two decades. Detecting multipartite entanglement, especially genuine multipartite entanglement of quantum systems is becoming a fundamental issue, due to its various applications in quantum information science \cite{Nilsen}.
A multipartite quantum state is called genuine multipartite entangled if it is not separable with respect to any bipartition \cite{Guhne}.
This special type of multipartite entanglement plays an important role in various quantum information processing tasks such as in the context of extreme spin squeezing \cite{extreme} and highly sensitive metrological tasks \cite{Hyllus,multi}.
It is also the basic ingredient in the measurement-based quantum computation \cite{Briegel} and in various quantum communication protocols \cite{Gisin,Raussendorf,Zhao,Yeo,Chen}.
However, characterization and detection of quantum entanglement is a formidably difficult task, and no efficient methods have been developed so far.
Researchers have devoted much to detect quantum entanglement \cite{GaoT1,GaoT2,Sperling,Eltschka,Yusef1,Yusef2,Zhanglin,Yusef3,Yusef4,Fisher5,Fisher6}.
As a special class of multipartite entanglement, genuine multipartite entanglement also attracts researchers' attention.
To better detect genuine multipartite entanglement, a series criteria have been presented such as linear and nonlinear entanglement witnesses
\cite{Huber1,Huber2,Huber4,Bancal,Huber5,Clivaz}, Bell-like inequalities \cite{Lim1},
and the norms of the correlation tensors \cite{Huber3,Lim2,Zhaohui}
and genuine multipartite entanglement concurrence \cite{Lim2,Ma1,Ma2,Lim3}.

In Ref. \cite{Fisher5}, the authors developed a method to detect bipartite entanglement by use of quantum Fisher information. They proposed an alternative entanglement criterion complementing to the criteria based on variance and local uncertainty relations. After then Akbari-Kourbolagh \emph{et al.} introduced another entanglement criterion for multipartite systems based on quantum Fisher information \cite{Fisher6}.

In this paper, inspired by the method in \cite{Fisher5}, we provide new criteria in detecting genuine tripartite entanglement based on quantum Fisher information.
We show that the new criteria are better than the existing ones by detailed example of tripartite states.

\section{Genuine entanglement criteria}
As we all know, there have been many genuine multipartite entanglement criteria.
Usually, researchers creat genuine multipartite entanglement criteria
by entanglement witness, the norms of the correlation tensors, genuine multipartite entanglement concurrence and so on.
Now we review some criterias.

In \cite{Huber5}, Clivaz {\it et al.} proposed a genuine multipartite entanglement criteria based on positive maps:
Let $\rho\in\mathcal{H}_1\otimes\mathcal{H}_2\otimes\cdots \mathcal{H}_n$,
and $A\subset\{1,2,\ldots,n\}$ denote a proper subset of the parties.
 A state $\rho_{2-sep}$ is biseparable if and only if it can be decomposed as
\begin{equation}
\rho_{2-sep}=\sum\limits_A\sum\limits_ip_i^{(A)}\rho_i^{(A)}\otimes\rho_i^{(\bar{A})}, \ p_i^{(A)}\geqslant0, \
\sum\limits_A\sum\limits_ip_i^{(A)}=1,
\end{equation}
where $\rho_A$ denotes a quantum state for the subsystem defined by the subset $A$ and
$\sum\limits_A$ stands for the sum over all bipartitions $A|\bar{A}$.
Then they sought for maps of the form
\begin{equation}
\Phi_{GME}:=\sum_A\Lambda_A\otimes \mathbb{I}_{\bar{A}}\circ \mathcal{U}^{(A)}+M,
\end{equation}
where $M$ is a positive map, $\mathcal{U}^{(A)}[\rho]=\sum_ip_i^{(A)}U_i^{(A)}\rho\left(U_i^{(A)}\right)^{\dag}$
is a family of convex combinations of local unitaries, and
$\Phi_{GME}[\rho_{2-sep}]\geqslant0$ for any $\rho_{2-sep}$.

In \cite{Lim3}, Li {\it et al.} gave two methods to detect genuine tripartite entanglement.
First, they proposed a criteria by the norms of the correlation tensors.
That is, $\rho$ is  genuine tripartite entangled if
\begin{equation}
M_k(\rho)>\frac{2\sqrt{2}}{3}(2\sqrt{k}+1)\frac{d-1}{d}\sqrt{\frac{d+1}{d}}, \ \ \ \forall k=1,2,\ldots,d^2-1,
\end{equation}
where $M_k(\rho)=\frac{1}{3}\left(\|T_{\underline{1}23}\|_k+\|T_{\underline{2}13}\|_k+\|T_{\underline{3}12}\|_k\right)$,
$\|M\|_k=\sum\limits_{i=1}^k\eta_i$ denote the $k$ norm for an $n\times n$ matrix $M$ with
$\eta_i$, $i=1,2,\ldots,n$ are singular values of $M$ in decreasing order,
and $T_{\underline{i}jk}$ be the corresponding correlation matrix for $\rho$.
Then they provided another method via genuine multipartite entanglement concurrence $C_{GME}$.
For a tripartite qudit state $\rho$,
\begin{equation}\label{concurrence}
C_{GME}\geqslant\max\left\{\frac{1}{2\sqrt{2}}\|T^{(123)}\|-\frac{d-1}{d},0\right\}.
\end{equation}
From \eqref{concurrence}, one can see $\rho$ is genuine entangled if $\frac{1}{2\sqrt{2}}\|T^{(123)}\|-\frac{d-1}{d}>0$.

Now we first briefly introduce some basic concepts about quantum Fisher information.
The quantum Fisher information $F(\rho,A)$ of a state $\rho$ with respect to an observable $A$ is defined by \cite{Fisher,Fisher1,Fisher2}
\begin{equation}\label{fisher1}
F(\rho,A)=\frac{1}{4}\tr\rho L^2,
\end{equation}
where $L$ is the symmetric logarithmic derivative determined by
$$
i[\rho,A]=\frac{1}{2}(L\rho+\rho L),
$$
with the square bracket denoting the commutator.

When the spectral decomposition of $\rho$ is known,
\begin{equation}
\rho=\sum\limits_k\lambda_k|k\rangle\langle k|,
\end{equation}
where $\lambda_k$ are the non-negative eigenvalues and $|k\rangle$ are the corresponding eigenvectors of $\rho$,
then for any observable $A$ on the system Hilbert space, the quantum Fisher information of \eqref{fisher1} can be
\begin{equation}
F(\rho,A)=\sum\limits_{k,l}\frac{(\lambda_k-\lambda_l)^2}{2(\lambda_k+\lambda_l)}|\langle k|A|l\rangle|^2,
\end{equation}
where the sums run over only those indices for which $\lambda_k+\lambda_l$ is nonzero \cite{Hyllus,Fisher2}.

The quantum Fisher information has the following remarkable information-theoretic properties \cite{Fisher3,Fisher4}:
{\rm (1)} Additivity:
$$
F(\rho^a\otimes \rho^b,A\otimes\mathbb{I}^b+\mathbb{I}^a\otimes B)=F(\rho^a,A)+F(\rho^b,B),
$$
where $\rho^a$ and $\rho^b$ are the local quantum states associated with the subsystems $a$ and $b$, $A$ and $B$ are observables,
and $\mathbb{I}^a$ and $\mathbb{I}^b$ stand for the identity operators on subsystems $a$ and $b$, respectively.

{\rm (2)} Convexity:
$$
F\left(\sum_j\lambda_j\rho_j,A\right)\leqslant\sum_j\lambda_jF(\rho_j,A)
$$
for quantum states $\rho_j$, where $\sum_j\lambda_j=1$, $\lambda_j\geqslant0$.

{\rm (3)} For any pure state $\rho$,
\begin{equation}\label{fisher2}
F(\rho,A)=(\Delta A)_\rho^2,
\end{equation}
where $(\Delta A)_\rho^2=\langle A^2\rangle_\rho-\langle A\rangle_\rho^2$ is the variance (uncertainty) of the observable $A$ with respect to the state $\rho$.

{\rm (4)} For an $N$-qudit quantum pure state $|\psi\rangle$ mixed with the white noise,
$\rho=p|\psi\rangle\langle\psi|+(1-p)\frac{I}{d^N}$,
\begin{equation}\label{fisher-1}
F(\rho,A)=\frac{p^2}{p+2(1-p)d^{-N}}F(|\psi\rangle,A).
\end{equation}

As the formula \eqref{fisher-1} is used many times in our paper, we give a brief proof for it.
Assume the spectral decomposition of $|\psi\rangle\langle\psi|$ is
\begin{equation*}
|\psi\rangle\langle\psi|=\sum\limits_k\lambda_k|k\rangle\langle k|
\end{equation*}
with $\lambda_1=1$, and $\lambda_2=\cdots\lambda_{d^N}=0$.
Then the corresponding spectral decomposition of $\rho$ can be
\begin{equation*}
\rho=\sum\limits_{k}\lambda_{k}^{\prime}|k\rangle\langle k|
\end{equation*}
with $\lambda_1^{\prime}=p+\frac{1-p}{d^{N}}$, and $\lambda_2^{\prime}=\cdots\lambda_{d^N}^{\prime}=\frac{1-p}{d^{N}}$.
Thus,
\begin{equation*}
F(\rho,A)
=\sum\limits_{k=2}^{d^{N}}\frac{(\lambda_1^{\prime}-\lambda_{k}^{\prime})^2}{2(\lambda_{1}^{\prime}+\lambda_{k}^{\prime})}|\langle k|A|l\rangle|^2
=\sum\limits_{k=2}^{d^{N}}\frac{p^2}{2(p+2(1-p)d^{-N})}|\langle k|A|l\rangle|^2
=\frac{p^2}{p+2(1-p)d^{-N}}F(|\psi\rangle,A).
\end{equation*}

We first present the following lemmas.

\begin{lemma}\label{FISHER1}
For any qubit state $\rho$,
$\sum\limits_{i=1}^3F(\rho,\sigma_i)\leqslant2$, where $\sigma_i$ are Pauli matrices.
\end{lemma}

\begin{proof}
Any qubit state $\rho$ can be written as $\rho=\frac{1}{2}(I+\vec{r}\cdot\vec{\mathbf{\sigma}})=
\frac{1}{2}(I+r_1\sigma_1+r_2\sigma_2+r_3\sigma_3)$,
where $\vec{r}$ is the 3-dimensional Bloch vector with $|\vec{r}|\leq 1$.
For a qubit pure state $|\psi\rangle$, $|\vec{r}|=1$, we have
$\sum\limits_{i=1}^3F(|\psi\rangle,\sigma_i)=3-|\vec{r}|^2=2$.
Hence, for any qubit state $\rho=\sum_j\lambda_j|\psi_j\rangle\langle\psi_j|$,
we obtain $\sum\limits_{i=1}^3F(\rho,\sigma_i)\leqslant\sum\limits_{j}\lambda_i\sum\limits_{i=1}^3F(|\psi_j\rangle,\sigma_i)=2$.
\end{proof}

\begin{lemma}\label{FISHER2}
For any bipartite qubit state $\rho^{ab}$,
\begin{equation}\label{lemma2}
\sum\limits_{i=1}^3F\left(\rho^{ab},\sigma_i^a\otimes \mathbb{I}^b+\mathbb{I}^a\otimes\sigma_i^b\right)\leqslant8,
\end{equation}
where $\sigma_i^a$ and $\sigma_i^b$ are Pauli matrices on subsystems $a$ and $b$, respectively.
\end{lemma}

\begin{proof}
By the convexity of quantum Fisher information, we only need to prove that the inequality holds for pure states.
Any bipartite qubit state $\rho^{ab}$ can be expressed as
$\rho^{ab}=\frac{1}{4}(I\otimes I+\vec{r}\cdot\vec{\sigma}\otimes I+I\otimes \vec{s}\cdot \vec{\sigma}+\sum_{i,j=1}^3t_{ij}\sigma_i\otimes\sigma_j)$,
where $\vec{r}=(r_1,r_2,r_3)^{T}$ and $\vec{s}=(s_1,s_2,s_3)^{T}$ are real 3-dimensional vectors, with $T$ denoting transpose.
Since the rank of a bipartite pure state $|\psi\rangle\langle\psi|$ is one, one has
$$
(t_{11}+t_{22})^2+(t_{12}-t_{21})^2=(1-t_{33})^2-(r_3-s_3)^2.
$$
Therefore, we have
$$
t_{11}+t_{22}=\pm\sqrt{(1-t_{33})^2-(r_3-s_3)^2-(t_{12}-t_{21})^2}.
$$
From Eq. \eqref{fisher2} it also holds that
$$
\sum\limits_{i=1}^3F\left(|\psi\rangle\langle\psi|,\sigma_i^a\otimes \mathbb{I}^b+\mathbb{I}^a\otimes\sigma_i^b\right)
=6+2\sum\limits_{i=1}^3t_{ii}-\sum\limits_{i=1}^3(r_i+s_i)^2.
$$
Hence, we get
$$
\begin{array}{rl}
& \ \ \ \ \sum\limits_{i=1}^3F\left(|\psi\rangle\langle\psi|,\sigma_i^a\otimes \mathbb{I}^b+\mathbb{I}^a\otimes\sigma_i^b\right)\\[1mm]
& \leqslant6+2\sum\limits_{i=1}^3t_{ii}\\[1mm]
& \leqslant6+2t_{33}+2\sqrt{(1-t_{33})^2-(r_3-s_3)^2-(t_{12}-t_{21})^2}\\[1mm]
& \leqslant6+2t_{33}+2\sqrt{(1-t_{33})^2}=8.
\end{array}
$$
The equality in (\ref{lemma2}) holds for $t_{12}=t_{21}$, $r_3=s_3=0$, and $r_1+s_1=r_2+s_2=0$,
for instance, $t_{11}=-1$, $t_{22}=t_{33}=1$, and the rest parameters are zero.
\end{proof}

Now we generalize Lemmas \ref{FISHER1} and \ref{FISHER2} to qudit and 2-qudit states conditions by Gell-Mann matrices, respectively.
The Gell-Mann matrices are defined as
\begin{equation}\label{gell-1}
\sigma_t^{jk}=|j\rangle\langle k|+|k\rangle\langle j|, \ \ 0\leqslant j<k\leqslant d-1,
\end{equation}
\begin{equation}\label{gell-2}
\sigma_s^{jk}=-i|j\rangle\langle k|+i|k\rangle\langle j|, \ \ 0\leqslant j<k\leqslant d-1,
\end{equation}
and
\begin{equation}\label{gell-3}
\sigma^{l}=\sqrt{\frac{2}{l(l+1)}}\left(\sum\limits_{j=0}^{l-1}|j\rangle\langle j|-l|l\rangle\langle l|\right), \ \ 1\leqslant l\leqslant d-1.
\end{equation}
\begin{lemma}\label{FISHER3}
For any $d$-dimensional qudit state $\rho$,
\begin{equation}
\sum\limits_{0\leqslant j<k\leqslant d-1}F\left(\rho,\sigma_t^{jk}\right)+\sum\limits_{0\leqslant j<k\leqslant d-1}F\left(\rho,\sigma_s^{jk}\right)
+\sum\limits_{1\leqslant l \leqslant d-1}F\left(\rho,\sigma^{l}\right)\leqslant2(d-1),
\end{equation}
where $\sigma_a^{jk}, \sigma_s^{jk}$, and $\sigma^{l}$ are Gell-Mann matrices
defined in \eqref{gell-1}, \eqref{gell-2}, and \eqref{gell-3}.
\end{lemma}
\begin{proof}
By the convexity of quantum Fisher information, we only need to prove the inequality holds for pure states.
Any $d$-dimensional pure state can be expressed as
$|\varphi\rangle=\sum\limits_{s=0}^{d-1}\varphi_s|s\rangle$ with $\sum\limits_{s=0}^{d-1}|\varphi_s|^2=1$.
Then one has
\begin{equation}
\sum\limits_{0\leqslant j<k\leqslant d-1}F\left(|\varphi\rangle,\sigma_t^{jk}\right)=
\sum\limits_{0\leqslant j<k\leqslant d-1}|\varphi_j|^2+|\varphi_k|^2-\left(\varphi_j\varphi_k^*+\varphi_j^*\varphi_k\right),
\end{equation}
\begin{equation}
\sum\limits_{l\leqslant l \leqslant d-1}\sum\limits_{0\leqslant j<k\leqslant d-1}F\left(|\varphi\rangle,\sigma_s^{jk}\right)=
\sum\limits_{0\leqslant j<k\leqslant d-1}|\varphi_j|^2+|\varphi_k|^2+\left(\varphi_j\varphi_k^*-\varphi_j^*\varphi_k\right),
\end{equation}
and
\begin{equation}
\begin{array}{rl}
& \ \ \ \ \sum\limits_{0\leqslant j<k\leqslant d-1}
F\left(|\varphi\rangle,\sigma^{l}\right)\\
&=\sum\limits_{0\leqslant j<k\leqslant d-1}\frac{2}{l(l+1)}\left(\left(\sum\limits_{j=0}^{l-1}|\varphi_j|^2-l|\varphi_l|^2\right)-
\left(\sum\limits_{j=0}^{l-1}|\varphi_j|^2+l^2|\varphi_l|^2\right)^2\right)\\
&=\sum\limits_{0\leqslant j<k\leqslant d-1}4|\varphi_j|^2|\varphi_k|^2.
\end{array}
\end{equation}
Thus,
\begin{equation}
\sum\limits_{0\leqslant j<k\leqslant d-1}F\left(|\varphi\rangle,\sigma_t^{jk}\right)+\sum\limits_{0\leqslant j<k\leqslant d-1}F\left(|\varphi\rangle,\sigma_s^{jk}\right)
+\sum\limits_{1\leqslant l \leqslant d-1}F\left(|\varphi\rangle,\sigma^{l}\right)= 2(d-1).
\end{equation}
\end{proof}

\begin{lemma}
For any 2-qudit state $\rho^{ab}\in\mathcal{H}_{ab}$ with $dim(\mathcal{H}_a)=dim(\mathcal{H}_b)=d$,
\begin{equation}\label{2-qudit}
\begin{array}{rl}
& \ \ \ \ \ \sum\limits_{0\leqslant j<k\leqslant d-1}
F\left(\rho^{ab},\left(\sigma_t^{jk}\right)^{a}\otimes\mathbb{I}^{b}+\mathbb{I}^{a}\otimes\left(\sigma_t^{jk}\right)^{b}\right)+
\sum\limits_{0\leqslant j<k\leqslant d-1}
F\left(\rho^{ab},\left(\sigma_s^{jk}\right)^{a}\otimes\mathbb{I}^{b}+\mathbb{I}^{a}\otimes\left(\sigma_s^{jk}\right)^{b}\right)\\
& \ \ +
\sum\limits_{1\leqslant l\leqslant d-1}
F\left(\rho^{ab},\left(\sigma^{l}\right)^{a}\otimes\mathbb{I}^{b}+\mathbb{I}^{a}\otimes\left(\sigma^{l}\right)^{b}\right)
\leqslant\frac{4(d-1)(d+2)}{d},
\end{array}
\end{equation}
where $\sigma_a^{jk}, \sigma_s^{jk}$, and $\sigma^{l}$ are Gell-Mann matrices
defined in \eqref{gell-1}, \eqref{gell-2}, and \eqref{gell-3}.
\end{lemma}
\begin{proof}
By the convexity of quantum Fisher information, we only need to prove the inequality holds for pure states.
Any pure state $|\varphi\rangle\in\mathcal{H}_{ab}$ can be expressed as
$|\varphi\rangle=\sum\limits_{m,n=0}^{d-1}\varphi_{mn}|mn\rangle$ with $\sum\limits_{m,n=0}^{d-1}|\varphi_{mn}|^2=1$.
Then one has
\begin{equation}
\begin{array}{rl}
& \ \ \ \ \sum\limits_{0\leqslant j<k\leqslant d-1}
F\left(|\varphi\rangle,\left(\sigma_t^{jk}\right)^{a}\otimes\mathbb{I}^{b}+\mathbb{I}^{a}\otimes\left(\sigma_t^{jk}\right)^{b}\right)\\
&=\sum\limits_{0\leqslant j<k\leqslant d-1}
\left(\sum\limits_{t=0}^{d-1}\left(|\varphi_{jt}|^2+|\varphi_{tj}|^2+|\varphi_{kt}|^2+|\varphi_{tk}|^2\right)
+2\left(\varphi_{jj}^*\varphi_{kk}+\varphi_{jj}\varphi_{kk}^*+\varphi_{jk}^*\varphi_{kj}+\varphi_{kj}^*\varphi_{jk}\right)\right)\\
& \ \ \ \ -\sum\limits_{0\leqslant j<k\leqslant d-1}\left[\sum\limits_{t=0}^{d-1}\left(\varphi_{jt}^*\varphi_{kt}+
\varphi_{jt}\varphi_{kt}^*+\varphi_{tj}^*\varphi_{tk}+\varphi_{tj}\varphi_{tk}^*\right)\right]^2\\
&=2(d-1)\\
& \ \ \ \ +2\sum\limits_{0\leqslant j<k\leqslant d-1}\left(
\left(\varphi_{jj}^*\varphi_{kk}+\varphi_{jj}\varphi_{kk}^*+\varphi_{jk}^*\varphi_{kj}+\varphi_{kj}^*\varphi_{jk}\right)
-\left[\sum\limits_{t=0}^{d-1}\left(\varphi_{jt}^*\varphi_{kt}+
\varphi_{jt}\varphi_{kt}^*+\varphi_{tj}^*\varphi_{tk}+\varphi_{tj}\varphi_{tk}^*\right)\right]^2\right),
\end{array}
\end{equation}
\begin{equation}
\begin{array}{rl}
& \ \ \ \ \sum\limits_{0\leqslant j<k\leqslant d-1}
F\left(|\varphi\rangle,\left(\sigma_s^{jk}\right)^{a}\otimes\mathbb{I}^{b}+\mathbb{I}^{a}\otimes\left(\sigma_s^{jk}\right)^{b}\right)\\
&=\sum\limits_{0\leqslant j<k\leqslant d-1}\left(
\sum\limits_{t=0}^{d-1}\left(|\varphi_{jt}|^2+|\varphi_{tj}|^2+|\varphi_{kt}|^2+|\varphi_{tk}|^2\right)
+2\left(\varphi_{jk}^*\varphi_{kj}+\varphi_{kj}^*\varphi_{jk}-\varphi_{jj}^*\varphi_{kk}-\varphi_{jj}\varphi_{kk}^*\right)\right)\\
& \ \ \ \ -\sum\limits_{0\leqslant j<k\leqslant d-1}
\left[\sum\limits_{t=0}^{d-1}\left(\left(i\varphi_{kt}^*\varphi_{jt}-i\varphi_{kt}\varphi_{jt}^*\right)
+\left(i\varphi_{tk}^*\varphi_{tj}-i\varphi_{tk}\varphi_{tj}^*\right)\right)\right]^2\\
&=2(d-1)\\
& \ \ \ \ +\sum\limits_{0\leqslant j<k\leqslant d-1}\left(
2\left(\varphi_{jk}^*\varphi_{kj}+\varphi_{kj}^*\varphi_{jk}-\varphi_{jj}^*\varphi_{kk}-\varphi_{jj}\varphi_{kk}^*\right)
-\left[\sum\limits_{t=0}^{d-1}\left(i\left(\varphi_{kt}^*\varphi_{jt}-\varphi_{kt}\varphi_{jt}^*\right)
+i\left(\varphi_{tk}^*\varphi_{tj}-\varphi_{tk}\varphi_{tj}^*\right)\right)\right]^2\right),
\end{array}
\end{equation}
and
\begin{equation}
\begin{array}{rl}
& \ \ \ \ \sum\limits_{1\leqslant l\leqslant d-1}
F\left(\rho,\left(\sigma^{l}\right)^{a}\otimes\mathbb{I}^{b}+\mathbb{I}^{a}\otimes\left(\sigma^{l}\right)^{b}\right)\\
&=\frac{4d-8}{d}+4\sum\limits_{m=0}^{d-1}|\varphi_{mm}|^2-\sum\limits_{1\leqslant l\leqslant d-1}
\frac{2}{l(l+1)}\left[\sum\limits_{t=0}^{d-1}\sum\limits_{j=0}^{l-1}\left(|\varphi_{jt}|^2+
|\varphi_{tj}|^-l|\varphi_{lt}|^2-l|\varphi_{tl}|^2\right)\right]^2.
\end{array}
\end{equation}
Thus,
\begin{equation}
\begin{array}{rl}
& \ \ \ \ \ \sum\limits_{0\leqslant j<k\leqslant d-1}
F\left(|\varphi\rangle,\left(\sigma_t^{jk}\right)^{a}\otimes\mathbb{I}^{b}+\mathbb{I}^{a}\otimes\left(\sigma_t^{jk}\right)^{b}\right)+
\sum\limits_{0\leqslant j<k\leqslant d-1}
F\left(|\varphi\rangle,\left(\sigma_s^{jk}\right)^{a}\otimes\mathbb{I}^{b}+\mathbb{I}^{a}\otimes\left(\sigma_s^{jk}\right)^{b}\right)\\
& \ \ +
\sum\limits_{1\leqslant l\leqslant d-1}
F\left(|\varphi\rangle,\left(\sigma^{l}\right)^{a}\otimes\mathbb{I}^{b}+\mathbb{I}^{a}\otimes\left(\sigma^{l}\right)^{b}\right)
\leqslant\frac{4(d-1)(d+2)}{d}.
\end{array}
\end{equation}
\end{proof}
Here, we note that the equality in \eqref{2-qudit} can hold when $\varphi_{mm}=\frac{1}{\sqrt{d}}$ for $i=0,1,\ldots,d-1$.

Now consider tripartite states $\rho_{abc}$ in systems $a$, $b$ and $c$.
Let $\{A_{\mu}\}$, $\{B_{\mu}\}$ and $\{C_{\mu}\}$ be the sets of local observables with respect to subsystems $a$, $b$ and $c$, respectively.
From Ref. \cite{Fisher5}, there exists $F_X$ and $F_{XY}$ such that
\begin{equation}\label{fisher4}
\sum_\mu F(\rho^X,X_{\mu})\leqslant F_{X},
\end{equation}
for any reduced local state $\rho^X$, where $X$ stands for any subsystem $a$, $b$ or $c$,
and
\begin{equation}\label{fisher9}
\sum_\mu F(\rho^{XY},A_{\mu}\otimes \mathbb{I}^Y+\mathbb{I}^X\otimes B_{\mu})\leqslant F_{XY},
\end{equation}
where $\rho^{XY}$ stand for reduced state associated with the subsystems $XY$, where $XY\in\{ab,ac,bc\}$.
Here, we note that $F_X$ and $F_{XY}$ are only depend on the local observable $X$ and $XY$, respectively.

Though there are different criteria, such as entanglement witness and other methods,
one often considers the genuine entanglement criteria based on biseparable state.

Let $\mathcal{H}_X$ denote the Hilbert space of the systems $X$. Consider tripartite
states $\rho^{abc}$ in $\mathcal{H}_a\otimes\mathcal{H}_b\otimes\mathcal{H}_c$ with $\dim\mathcal{H}_a=\dim \mathcal{H}_b=\dim \mathcal{H}_c$.
$\rho^{abc}$ is said to be genuine entangled if it cannot be written in the following form,
\begin{equation}\label{bs}
\rho^{abc}=\sum\limits_ip_i\rho_i^a\otimes\rho_i^{bc}+\sum\limits_jq_j\rho_j^b\otimes\rho_j^{ac}
+\sum\limits_lr_l\rho_l^{ab}\otimes\rho_l^c.
\end{equation}
A state of the form (\ref{bs}) is called bi-separable. For a  bi-separable state, it can be verified that
$F_a=F_b=F_c$ and $F_{ab}=F_{ac}=F_{bc}$.

\begin{theorem}
A tripartite state $\rho^{abc}\in\mathcal{H}_a\otimes\mathcal{H}_b\otimes\mathcal{H}_c$ with $\mathcal{H}_a=\mathcal{H}_b=\mathcal{H}_c=d$
is genuine entangled if
\begin{equation}\label{thm-abc}
\sum\limits_{\mu}F(\rho^{abc},A_\mu\otimes\mathbb{I}^{bc}+B_\mu\otimes\mathbb{I}^{ac}+\mathbb{I}^{ab}\otimes C_\mu)> F_1+F_2,
\end{equation}
where $F_1=F_a=F_b=F_c$, and $F_2=F_{ab}=F_{ac}=F_{bc}$.
\end{theorem}

\begin{proof}
By the additivity and convexity of quantum Fisher information, if $\rho$ is defined as \eqref{bs}, we have
$$
\begin{array}{rl}
& \ \ \ \ \ \ \sum\limits_{\mu}F(\rho^{abc},A_{\mu}\otimes\mathbb{I}^{bc}+B_{\mu}\otimes\mathbb{I}^{ac}+\mathbb{I}^{ab}\otimes C_{\mu})\\
& \ \ \leqslant \sum\limits_{\mu}\sum\limits_ip_i\sum F(\rho_i^a\otimes\rho_i^{bc},A_{\mu}\otimes\mathbb{I}^{bc}+
B_{\mu}\otimes\mathbb{I}^{ac}+\mathbb{I}^{ab}\otimes C_{\mu})\\
& \ \ \ \ +\sum\limits_{\mu}\sum\limits_jq_j\sum F(\rho_j^b\otimes\rho_j^{ac},A_{\mu}\otimes\mathbb{I}^{bc}+
B_{\mu}\otimes\mathbb{I}^{ac}+\mathbb{I}^{ab}\otimes C_{\mu})\\
& \ \ \ \ +\sum\limits_{\mu}\sum\limits_lr_l\sum F(\rho_j^{ab}\otimes\rho_j^{c},A_{\mu}\otimes\mathbb{I}^{bc}+
B_{\mu}\otimes\mathbb{I}^{ac}+\mathbb{I}^{ab}\otimes C_{\mu})\\
& \ \ = \sum\limits_ip_i\sum\limits_{\mu}(F(\rho_i^a,A_{\mu})+F(\rho_i^{bc},B_{\mu}\otimes\mathbb{I}^{c}+\mathbb{I}^{b}\otimes C_{\mu}))
+\sum\limits_jq_j\sum\limits_{\mu}(F(\rho_j^b,B_{\mu})+F(\rho_j^{ac},A_{\mu}\otimes\mathbb{I}^{c}+\mathbb{I}^{a}\otimes C_{\mu}))\\
& \ \ \ \ +\sum\limits_lr_l\sum\limits_{\mu}(F(\rho_l^c,C_{\mu})+F(\rho_l^{ab},A_{\mu}\otimes\mathbb{I}^{b}+\mathbb{I}^{a}\otimes B_{\mu}))\\
& \ \ \leqslant \sum\limits_ip_i(F_A+F_{BC})+\sum\limits_jq_j(F_{B}+F_{AC})
+\sum\limits_lr_l(F_{c}+F_{ab})\\
& \ \ = F_1+F_2.
\end{array}
$$
\end{proof}

If we choose $A_\mu$, $B_\mu$, and $C_\mu$ in \eqref{thm-abc} as Gell-Mann matrices,
from Lemmas \ref{FISHER1}, \ref{FISHER2} one gets the following two corollaries.

\begin{corollary}\label{GHZ}
A tripartite state $\rho^{abc}\in\mathcal{H}_a\otimes\mathcal{H}_b\otimes\mathcal{H}_c$ with $\mathcal{H}_a=\mathcal{H}_b=\mathcal{H}_c=d$
is genuine entangled if
\begin{equation}
\begin{array}{rl}
& \ \ \ \ \sum\limits_{0\leqslant j<k\leqslant d-1}
F\left(\rho^{abc},\left(\sigma_t^{jk}\right)^{a}\otimes\mathbb{I}^{bc}+\mathbb{I}^{ac}\otimes\left(\sigma_t^{jk}\right)^{b}
+\mathbb{I}^{ab}\otimes\left(\sigma_t^{jk}\right)^{c}\right)\\
&+
\sum\limits_{0\leqslant j<k\leqslant d-1}
F\left(\rho^{abc},\left(\sigma_s^{jk}\right)^{a}\otimes\mathbb{I}^{bc}+\mathbb{I}^{ac}\otimes\left(\sigma_s^{jk}\right)^{b}
+\mathbb{I}^{ab}\otimes\left(\sigma_s^{jk}\right)^{c}\right)\\
&+
\sum\limits_{1\leqslant l\leqslant d-1}
F\left(\rho^{abc},\left(\sigma^{l}\right)^{a}\otimes\mathbb{I}^{bc}+\mathbb{I}^{ac}\otimes\left(\sigma^{l}\right)^{b}
+\mathbb{I}^{ab}\otimes\left(\sigma^{l}\right)^{c}\right)\\
&> \frac{2(d-1)(3d+4)}{d},
\end{array}
\end{equation}
\end{corollary}

\begin{corollary}\label{FISHER}
For a three-qubit state $\rho^{abc}$, if
\begin{equation}\label{fisher12}
\sum\limits_{i=1}^3F(\rho^{abc},\pm\sigma_i\otimes\mathbb{I}^{bc}\pm\sigma_i\otimes\mathbb{I}^{ac}\pm\mathbb{I}^{ab}\otimes \sigma_i)>10,
\end{equation}
then $\rho^{ABC}$ is genuine entangled.
\end{corollary}

Let us consider the following examples.
\begin{example}
Consider the mixture of the three-qubit GHZ state and $W$ state,
\begin{equation}\label{tq}
\rho=\frac{1-x-y}{8}I+x|GHZ\rangle\langle GHZ|+y|W\rangle\langle W|,
\end{equation}
where $|GHZ\rangle=\frac{1}{\sqrt{2}}(|000\rangle+|111\rangle)$, and $|W\rangle=\frac{1}{\sqrt{3}}(|100\rangle+|010\rangle+|001\rangle)$.
Set $A_1=B_1=C_1=\sigma_1$, $A_2=B_2=C_2=\sigma_2$, and $A_3=B_3=-C_3=\sigma_3$.
Then we have
$$
\begin{array}{rl}
& \ \ \ \ \ \ F(\rho,A_1\otimes\mathbb{I}^{bc}+B_1\otimes\mathbb{I}^{ac}+\mathbb{I}^{ab}\otimes C_1)\\[1mm]
& \ \ =F(\rho,A_2\otimes\mathbb{I}^{bc}+B_2\otimes\mathbb{I}^{ac}+\mathbb{I}^{ab}\otimes C_2)\\[1mm]
& \ \ =\displaystyle\frac{6x^2}{1+3x-y}+\frac{22y^2}{1+3y-x}+\frac{6(x-y)^2}{1+3x+3y},
\end{array}
$$
and
$$
\begin{array}{rl}
& \ \ \ \ \ F(\rho,A_3\otimes\mathbb{I}^{bc}+B_3\otimes\mathbb{I}^{ac}+\mathbb{I}^{ab}\otimes C_3)\\[1mm]
& \ \ =\displaystyle\frac{4x^2}{1+3x-y}+\frac{128y^2}{9(1+3y-x)}.
\end{array}
$$
Thus
$$
\begin{array}{rl}
& \ \ \ \ \ \ \sum\limits_{\mu}F(\rho,A_{\mu}\otimes\mathbb{I}^{bc}+B_{\mu}\otimes\mathbb{I}^{ac}+\mathbb{I}^{ab}\otimes C_{\mu})\\
& \ \ =\displaystyle\frac{16x^2}{1+3x-y}+\frac{524y^2}{9(1+3y-x)}+\frac{12(x-y)^2}{1+3x+3y}.
\end{array}
$$
Denote $f(x,y)=\frac{16x^2}{1+3x-y}+\frac{524y^2}{9(1+3y-x)}+\frac{12(x-y)^2}{1+3x+3y}-10$.
From Corollary 1 the three-qubit state (\ref{tq}) is genuine entangled if $f(x,y)>0$, see Fig. 1.

\begin{figure}
  \centering
  \includegraphics[width=6cm]{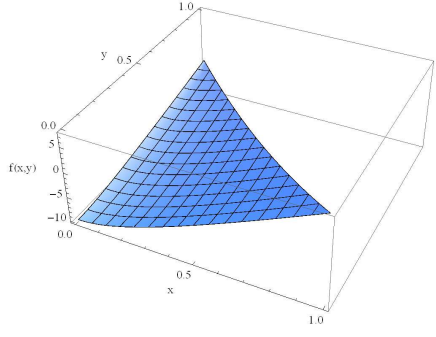}
  \caption{$f(x,y)$ v.s. $x$ and $y$. Here, $f(x,y)$ stands for the function defined in Example 1.}
\end{figure}

We can see that for some states our method is more efficient.
For instance, take $x=0$. Then
\begin{equation}\label{tq0}
\rho=\frac{1-y}{8}I+y|W\rangle\langle W|, \ 0\leqslant y\leqslant1
\end{equation}
is genuine entangled when $y>0.647236$, see Fig. 2.

\begin{figure}
  \centering
  \includegraphics[width=6cm]{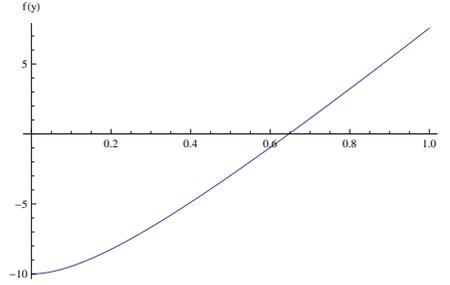}
  \caption{From Lemma \ref{FISHER}, $\rho$ is genuine entangled if $f(0,y)=\frac{632y^2}{9(3y+1)}-10>0$.
  From the figure above, one can see that it means when $y>0.647236$, $\rho$ is genuine entangled.}
\end{figure}

It has been obtained in \cite{Clivaz} and \cite{Lim3} that the state (\ref{tq0}) is genuine entangled for $y>0.90$ and $y>0.738549$, respectively. Obviously, our criterion is better
than the one presented in \cite{Clivaz} and \cite{Lim3} for detecting the genuine entanglement of the state defined in (\ref{tq0}).
Thus, better than Vicente criterion by Theorem 2 in \cite{Huber3},
since the criteria proposed in \cite{Lim3} is more efficient of the state defined in \eqref{tq0}.
\end{example}

\begin{example}
Consider quantum state $\rho\in\mathcal{H}_a\otimes\mathcal{H}_b\otimes\mathcal{H}_c$,
$\rho=\frac{1-p}{d^3}\mathbb{I}+p|GHZ\rangle\langle GHZ|$,
where $|GHZ\rangle=\frac{1}{\sqrt{d}}\sum\limits_{j=0}^{d-1}|jjj\rangle$  is the Greenberger-Horne-Zeilinger (GHZ) state.
Then one has
\begin{equation}
\sum\limits_{0\leqslant j<k\leqslant d-1}
F\left(|GHZ\rangle,\left(\sigma_t^{jk}\right)^{a}\otimes\mathbb{I}^{bc}+\mathbb{I}^{ac}\otimes\left(\sigma_t^{jk}\right)^{b}
+\mathbb{I}^{ab}\otimes\left(\sigma_t^{jk}\right)^{c}\right)=3(d-1),
\end{equation}
\begin{equation}
\sum\limits_{0\leqslant j<k\leqslant d-1}
F\left(|GHZ\rangle,\left(\sigma_s^{jk}\right)^{a}\otimes\mathbb{I}^{bc}+\mathbb{I}^{ac}\otimes\left(\sigma_s^{jk}\right)^{b}
+\mathbb{I}^{ab}\otimes\left(\sigma_s^{jk}\right)^{c}\right)=3(d-1),
\end{equation}
and
\begin{equation}
\sum\limits_{1\leqslant l\leqslant d-1}
F\left(|GHZ\rangle,\left(\sigma^{l}\right)^{a}\otimes\mathbb{I}^{bc}+\mathbb{I}^{ac}\otimes\left(\sigma^{l}\right)^{b}
+\mathbb{I}^{ab}\otimes\left(\sigma^{l}\right)^{c}\right)=\frac{18(d-1)}{d}.
\end{equation}
Thus, from \eqref{fisher-1}, one has
\begin{equation}
\begin{array}{rl}
& \ \ \ \ \sum\limits_{0\leqslant j<k\leqslant d-1}
F\left(\rho,\left(\sigma_t^{jk}\right)^{a}\otimes\mathbb{I}^{bc}+\mathbb{I}^{ac}\otimes\left(\sigma_t^{jk}\right)^{b}
+\mathbb{I}^{ab}\otimes\left(\sigma_t^{jk}\right)^{c}\right)\\
&+\sum\limits_{0\leqslant j<k\leqslant d-1}
F\left(\rho,\left(\sigma_s^{jk}\right)^{a}\otimes\mathbb{I}^{bc}+\mathbb{I}^{ac}\otimes\left(\sigma_s^{jk}\right)^{b}
+\mathbb{I}^{ab}\otimes\left(\sigma_s^{jk}\right)^{c}\right)\\
&+\sum\limits_{1\leqslant l\leqslant d-1}
F\left(\rho,\left(\sigma^{l}\right)^{a}\otimes\mathbb{I}^{bc}+\mathbb{I}^{ac}\otimes\left(\sigma^{l}\right)^{b}
+\mathbb{I}^{ab}\otimes\left(\sigma^{l}\right)^{c}\right)\\
&=\frac{d^3p}{(2+(d^3-2)p)}\left(6(d-1)+6(d-1)+\frac{18(d-1)}{d}\right)\\[2.00mm]
&=\frac{6p^2d^2(d-1)(d+3)}{2+(d^3-2)p}.
\end{array}
\end{equation}
Define $g(d,p)=\frac{6p^2d^2(d-1)(d+3)}{2+(d^3-2)p}-\frac{2(d-1)(3d+4)}{d}$.
By Corollary \ref{GHZ}, one can find $\rho$ is genuine entangled if $g(d,p)>0$,
i.e., $p>\frac{\sqrt{((4 + 3^d) (16 + 12^d + 56d^3 + 12 d^4 + 4 d^6 + 3 d^7))}+(4 + 3 d) (-2 + d^3)}{6d^3 (3 + d)}$.
For $d=2$, one can see $p>0.728714$ which is better than the criteria given in \cite{Clivaz}
since in \cite{Clivaz}, $p>\frac{11}{15}\approx 0.733333$.
\end{example}

Our criterion can be generalized to the general tripartite systems with different local dimensions, $\dim\mathcal{H}_a=d_1$, $\dim \mathcal{H}_b=d_2$ and $\dim \mathcal{H}_c=d_3$.
We have

\begin{theorem}
Any bi-separable tripartite state $\rho\in\mathcal{H}_a\otimes\mathcal{H}_b\otimes\mathcal{H}_c$ satisfies
\begin{equation}
\sum\limits_{\mu}F(\rho,A_{\mu}\otimes\mathbb{I}^{bc}+B_{\mu}\otimes\mathbb{I}^{ac}+\mathbb{I}^{ab}\otimes C_{\mu})\leqslant F_1+F_2,
\end{equation}
where $F_1=\max\{F_a,F_b,F_c\}$, and $F_2=\max\{F_{ab},F_{bc},F_{ac}\}$.
\end{theorem}

\begin{proof}
From the proof of Theorem 1, one has
\begin{equation}
\begin{array}{rl}
& \ \ \ \ \sum\limits_{\mu}F(\rho,A_{\mu}\otimes\mathbb{I}^{bc}+B_{\mu}\otimes\mathbb{I}^{ac}+\mathbb{I}^{ab}\otimes C_{\mu})\\
&\leqslant\sum\limits_ip_i(F_a+F_{bc})+\sum\limits_jq_j(F_{b}+F_{ac})+\sum\limits_lr_l(F_{c}+F_{ab})\\
&\leqslant F_1+F_2.
\end{array}
\end{equation}
\end{proof}
\section{Conclusion}
Detecting genuine multipartite entanglement is a fundamental and significant task in quantum information theory.
We have obtained a criteria to detect genuine tripartite entanglement based on quantum Fisher information.
Particularly, for three-qubit state systems, example shows that our criterion detects better
the genuine entanglement than the existing criterion. Moreover, we have generalized the results to any tripartite systems
with arbitrarily different local dimensions.

\section{Acknowledgements}
This work is supported by the National Natural Science Foundation of China under Grant Nos. 11805143 and 11675113,
Beijing Municipal Commission of Education (KZ201810028042) and Academy for Multidisciplinary Studies of Capital Normal University.

\end{document}